\documentclass[submission,copyright,creativecommons]{eptcs}
\usepackage{graphicx}
\usepackage{bussproofs}
\usepackage{amsmath}
\usepackage{amsthm}
\usepackage{amssymb}
\usepackage{wrapfig}
\usepackage{color}
\usepackage{lipsum}

\def\a{\AxiomC}
\def\u{\UnaryInfC}
\def\b{\BinaryInfC}

\newtheorem{definition}{Definition}

\newtheorem{lemma}{Lemma}
\newtheorem{theorem}{Theorem}

\newtheorem{remark}[theorem]{Remark} 
\newtheorem{proposition}{Proposition}

\title{Proof-graphs for Minimal Implicational Logic}
\author{Marcela Quispe-Cruz \institute{Inform\'atica PUC-Rio, Rio de Janeiro, Brazil}
\and Edward Hermann Haeusler \institute{Inform\'atica PUC-Rio, Rio de Janeiro, Brazil}
\and Lew Gordeev \institute{T\"{u}bingen\medskip\ University, Ghent University, PUC-Rio}
}
\def\titlerunning{Proof-graphs for minimal implicational logic}
\def\authorrunning{M. Quispe-Cruz \&  E.H. Haeusler \&  L. Gordeev}
\begin{document}
\maketitle

\begin{abstract} 
It is well-known that the size of propositional classical proofs can  be huge. 
Proof theoretical studies discovered exponential gaps between  normal or cut free 
proofs and their respective non-normal proofs. The  aim of this work is to study 
how to reduce the weight of propositional deductions. We present the formalism of 
\emph{proof-graphs} for purely implicational logic, which are graphs 
of a specific shape that are  intended to capture the logical structure of a deduction. 
The advantage of this formalism is that formulas can be shared in the reduced proof.

In the present paper we give a precise definition of proof-graphs for the  minimal
implicational logic, together with a normalization procedure for  these proof-graphs. 
In contrast to standard tree-like formalisms, our normalization does not increase the 
number of nodes, when applied to the corresponding minimal proof-graph representations. 
\end{abstract}

\section{Introduction}

The use of proof-graphs, 
instead of trees or lists, for representing proofs is getting popular among 
proof-theoreticians. Proof-graphs serve as a way to provide a better symmetry to the 
semantics of proofs \cite{deOliveira2003} and a way to study 
complexity of propositional proofs and to provide more efficient theorem 
provers, concerning size of propositional proofs. In \cite{BonetB1993}, 
one can find a complexity analysis of the size of Frege systems, 
Natural Deduction systems and Sequent Calculus concerning their tree-like and 
list-like representation. This leads to $O(nlog(n))$
improvement in the size of the list-based proofs compared to tree-like proofs, 
which is based on the observation that the hypotheses occur only
once in the lists and more than once in the trees. Thus sharing formulas helps 
to reduce the size of proofs. There are related works, e.g. \cite{Alves2011}, 
that use graphs for representing proofs, pointing out that proof-graphs offer 
a better way to facilitate the visualisation and understanding of proofs in 
the underlying logic.

On the other hand \cite{Finger2005}, \cite{Vaston07} and \cite{Gordeev2009}
show that the use of Directed Acyclic Graphs (DAGs) together with mechanisms
of unification/substitution in proof representations has
compacting/compressing factor equivalent to cut-introduction. And,
obviously, graphs can save space by means of reference, instead of plain
copying. This paper shows yet another advantage of using graphs for
representing proofs. We show that using ``mixed'' graph representations of
formulas and inferences in Natural Deduction in the purely implicational
minimal logic one can obtain a (weak) normalization theorem that, in fact,
is a strong normalization theorem. Moreover the corresponding normalization
procedure does not exceed the size of the input, which sharply contrasts to
the well-known exponential speed-up of standard normalization. The choice of
purely implicational minimal logic ($M^{\rightarrow }$) is motivated by the
fact that the computational complexity of the validity of $M^{\rightarrow }$
is PSPACE-complete and can polynomially simulate classical, intuitionistic
and full minimal logic \cite{Statman79} as well as any propositional
logic with a Natural Deduction system satisfying the subformula property  
\cite{Haeusler2013}.

In a more general context, this work has been conducted as part of a bigger
tree-to-graph proof compressing research project. The purpose of such proof
compression is:

\begin{enumerate}
\item  To construct small (if possible, minimal) graph-like representations
of standard tree-like proofs in a given proof system and -- in the
propositional case -- investigate the corresponding short graph-like theorem
provers.

\item  To find short (say, polynomial-size) graph-like analogous of
standard tree-like proof theoretic operations like e.g. normalization in
Natural Deduction and/or cut-elimination in Sequent Calculus.
\end{enumerate}

Note that the present work fulfills both conditions with regard to the mimp-graph
representation (see below) of chosen Natural Deduction and the corresponding
notion of formula-minimality (see Theorems \ref{theo1} and \ref{theo2}).

Back to the proof normalization, recall the following properties of a given
structural deductive system (Natural Deduction, Sequent Calculus, etc):

\begin{itemize}
\item Normal form: To each derivation of $\alpha$ from $\Delta$ there is 
a normal derivation of $\alpha$ from $\Delta^{\prime}\subseteq\Delta$.
\item Normalization: To each derivation of $\alpha$ from $\Delta$ there 
is a normal derivation of $\alpha$ from $\Delta^{\prime}\subseteq\Delta$, 
obtained by a particular strategy of reductions application. 
\item Strong Normalization: To each derivation of $\alpha$ from $\Delta$ 
there is a normal derivation of $\alpha$ from $\Delta^{\prime}\subseteq\Delta$. 
This normal form can be obtained by applying reductions to the original 
derivation in any ordering. 
\end{itemize}

The strong normalization property for a natural deduction system is usually
proved by the so-called semantical method:

\begin{itemize}
\item  Define a property $P(\pi)$ on derivations $\pi$ in the Natural
Deduction system;

\item  Prove that this property implies strong normalization, that is $%
\forall\pi (P(\pi)\rightarrow SN(\pi))$, where $SN(X)$ means that $X$ is
strongly normalizable;

\item  Prove that $\forall \pi P(\pi)$.
\end{itemize}

There are well-known examples of this property $P(X)$ : (1) Prawitz's
``strong validity''; (2) Tait's ``convertibility''; (3) Jervell's
``regularity''; (4) Leivant's ``stability''; (5) Martin-L\"{o}f's
``computability''; (6) Girard's ``candidate de reducibilit\'{e}''. Note that
such semantical method is inconstructive and even in the case of purely
implicational fragment of minimal logic it provides no combinatorial insight
into the nature of strong normalization. Another, more constructive strategy
would be to show that there is a worst sequence of reductions always
producing a normal derivation. Let us call it a syntactic method of proving
the strong normalization theorem. This method is used in the present paper.

Other methods use assignments of rather complicated measures to derivations
such that arbitrary reductions decrease the measure, which by standard
inductive arguments yields a desired proof of the strong normalization. In
this paper we show how to represent $M^{\rightarrow }$ derivations in a
graph-like form and how to reduce (eliminating maximal formulas) these
representations such that a normalization theorem can be proved by counting
the number of maximal formulas in the original derivation. The strong
normalization will be a direct consequence of such normalization, since any
reduction decreases the corresponding measure of derivation complexity. The
underlying intuition comes from the fact that our graph representations use
only one node for any two identical formulas occurring in the original
Natural Deduction derivation (see Theorem \ref{theo1} for a more precise description).

In \cite{Geuvers2007} another approach to represent Natural Deduction using 
graphs is proposed. It reports a graph-representation of Natural Deduction, in
Gentzen as well as Fitch's style. In fact the proofs are represented as
hypergraphs, or boxed-graphs, with possibility of sharing subproofs. It is
developed not only for the implicational fragment, although the
representation of linear logic proofs is related as further work. Our
approach is different from \cite{Geuvers2007}, in that we include
graph-representations of formulas in the proofs. The fact that our
normalization procedure leads to strong-normalization is a consequence of
sharing subformulas, and hence subproofs, in our proof graph
representations. It is unclear whether a similar result is available using 
\cite{Geuvers2007}. 

\section{Mimp-graphs}\label{sectionPG}

Mimp-graphs are special directed graphs whose nodes and
edges are assigned with labels. Moreover we distinguish between formula
nodes and rule nodes. The formula nodes are labelled with formulas as being
encoded/represented by their principal connectives (in particular, atoms)
and the rule nodes are labelled with the names of the inference rules 
($\rightarrow$I and $\rightarrow$E). Both logic connectives proper and
inference names may be indexed, in order to achieve a 1--1 correspondence
between formulas (inferences) and their representations (names). Since
formulas are uniquely determined by the representations in question, i.e.
formula node labels, in the sequel we'll sometimes identify both; to
emphasize the difference we'll refer to the formers as formula graphs, i.e.
the ones whose formula node labels are formulas, instead of principal
connectives. The edges are labelled with tokens that identify the
connections between the respective rule nodes and formula nodes. Note that
formulas may occur only once in the mimp-graph. Subformulas are indicated by
outgoing edges with labels $l$ (left) and $r$ (right), see Figure~\ref
{fornodes}.

\begin{figure}[h]
\center
\includegraphics[scale=0.67]{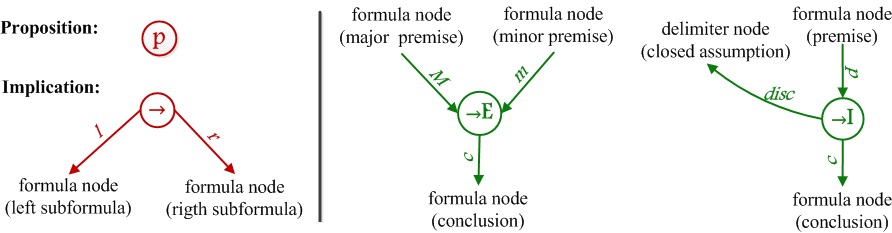} 
\caption{Types of formula nodes of the formula-graph and types of rule nodes of the mimp-graph}
\label{fornodes}
\end{figure}

The rule nodes, like in Natural Deduction, require the correct number of
premises. The premises are indicated by ingoing edges and there are edges
from the rule nodes to the conclusion formulas. The right-hand side of Figure~%
\ref{fornodes} shows the rule nodes $\rightarrow$I (implication
introduction) and $\rightarrow$E (implication elimination). Note that the
discharging of hypotheses may be vacuous. This case in a mimp-graph is
represented by a disconnected graph, where the discharged formula node is
not linked to the conclusion of the rule by any directed path.

\begin{figure}[h]
\begin{minipage}[c]{15cm}
\centering
	\begin{small}
	\def\defaultHypSeparation{\hskip .3in}
	\def\ScoreOverhang{0pt}
	  \a{$[p]^1$}
	  \a{$p \to q$}\RightLabel{ $\to$-$E$}
	  \b{$q$}\RightLabel{$\to$-$E$}
  	  \a{$[q \to r]^2$}
	  \b{$r$}\RightLabel{($\to$-$I$,{\scriptsize 1})}
	  \u{$p \to r$}\RightLabel{($\to$-$I$,{\scriptsize 2})}
	  \u{$(q \to r) \to (p \to r)$}
	 \DisplayProof
	\end{small}
	
		 \vspace{0.4cm}
 	$\Downarrow trans$
	 \vspace{0.5cm}
\end{minipage}

\begin{minipage}[r]{15cm}
	\centering
	\includegraphics[scale=0.65]{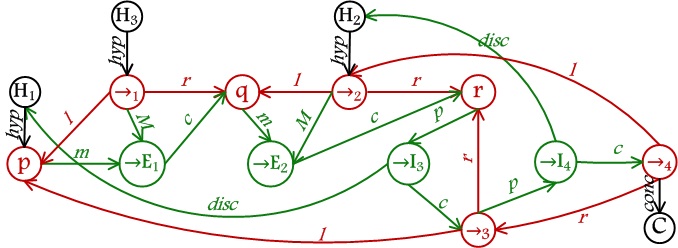}
\end{minipage}
\caption{The transition from a natural deduction proof to a mimp-graph}
\label{example1}
\end{figure}

\begin{figure}[h]
\begin{minipage}[c]{15cm}
\centering
\begin{small}
\def\defaultHypSeparation{\hskip .2in}
\def\ScoreOverhang{0pt}
\begin{prooftree}
\AxiomC{$[r]^{1}$}
\UnaryInfC{$((r \to s) \to r) \to r$}
\AxiomC{$[ (((r \to s) \to r) \to r) \to s]^{3}$}
\BinaryInfC{$s$} \RightLabel{1}
\UnaryInfC{$r\to s$} 
\AxiomC{$[(r \to s) \to r]^{2}$}
\BinaryInfC{$r$}\RightLabel{2}
\UnaryInfC{$(((r \to s) \to r) \to r)$}
\AxiomC{$[ (((r \to s) \to r) \to r) \to s]^{3}$}
\BinaryInfC{$s$}\RightLabel{3}
\UnaryInfC{$((((r \to s) \to r) \to r) \to s)\to s$}
\end{prooftree}

\end{small}
		 \vspace{0.2cm}
 	$\Downarrow trans$
	 \vspace{0.4cm}
\end{minipage}

\begin{minipage}[r]{15cm}
	\centering
	\includegraphics[scale=0.6]{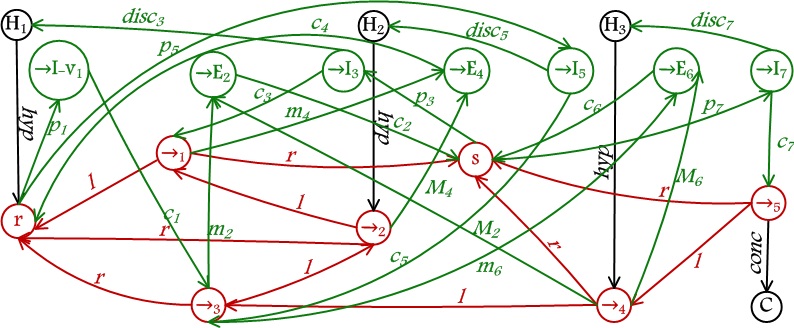}
\end{minipage}
\caption{The transition from the derivation of the formula $((((r \to s) \to r) \to r) \to s) \to s$ to a mimp-graph}
\label{example2}
\end{figure}

In the rule nodes, formulas are re-used, which is indicated by putting
several arrows towards it, hence the number of ingoing/outgoing edges with
label $p$ (premise), $M$ (major premise), $m$ (minor premise) and $c$
(conclusion) coming or going to a formula node could be arbitrarily large.
To make all this a bit more intuitive we give an example of a mimp-graph in
Figure~\ref{example1}, which can be seen as a derivation of $(q\rightarrow
r)\rightarrow (p\rightarrow r)$ from $(p\rightarrow q)$. Indices of
discarded hypotheses are replaced by additional edges assigned with the
label: $disc$ (discharge). This re-using of formulas is necessary. We remind
the reader that some valid implicational formulas, such as $((((r \to s) \to r) \to r) \to s) \to s$ (see Figure~\ref{example2}), 
need to use twice a subformula in a  Natural Deduction proof, in this case the subformula
$(((r \to s) \to r) \to r) \to s $ is used twice. Because of this, the edges $p$, $m$, $M$ and $c$ 
in Figure~\ref{example2} are indexed in a unique way.

The formula nodes in the graph (Figure~\ref{example1}) are labelled with
propositional letters $p$, $q$ and $r$, the connective $\rightarrow $; the
rule nodes are labelled with $\rightarrow$E and $\rightarrow$I. The
underlying idea is that there is an inferential order between rule nodes
that provides the corresponding derivability order; the formula node
labelled $\rightarrow _{4}$ linked to the delimiter node $C$ by an edge
labelled $conc$ is the root node and the conclusion of the proof represented
by the graph. Besides, the node $\rightarrow _{1}$ linked to the delimiter
node $H$ by the edge labelled $hyp$ (hypothesis) in the graph is
representing the premise $(p\rightarrow q)$.

We want to emphasize that the mimp-graphs put together information on
formula-graphs and rule nodes. To make it more transparent we can use
bicolored graphs. In this way formula nodes and edges between them are
painted red, whereas inference nodes and edges between them and adjacent
premises and/or conclusions are painted green. So nodes of types $\to$ and $%
p $ (propositions) together with adjacent edges $(l, r)$ are red, whereas
nodes labelled $\to$I and $\to$E together with adjacent edges $(m, M, p, c,
disc)$ are green. 

Now we give a formal definition of mimp-graphs.

\begin{definition}
\label{defLbl} \emph{L} is the union of the four sets of labels types:

\begin{itemize}
\item  \emph{R-Labels} is the set of inference labels: \{$\to$I$_n / n\in 
\mathbb{Z}\} \cup \{ \to$E$_m/ m \in \mathbb{Z}\}$,

\item  \emph{F-Labels} is the set of formula labels: \{$\to_i / i \in 
\mathbb{N}$\} and the propositional letters $\{p, q, r, . . .\}$,

\item  \emph{E-Labels} is the set of edge labels: \{$l$ (left), $r$ (right), 
 $conc$ (final conclusion), $hyp$ (hypothesis)$\}$ $ \cup $ $\{ p_j$ (premise)$/ j \in \mathbb{Z}\}$ $ \cup $ $ \{ m_j$ (minor premise)$/ j \in \mathbb{Z}\} $ $\cup $ $ \{ M_j$ (major premise)$/ j \in \mathbb{Z}\}$ $\cup$ $ \{ c_j$ (conclusion)$/ j \in \mathbb{Z}\}$ $ \cup$ $ \{disc_j$ (discharge)$/ j \in \mathbb{Z}\}$,

\item  \emph{D-Labels} is the set of delimiter labels: $\{H_k/ k \in \mathbb{%
Z}\} \cup \{C\}$.
\end{itemize}
\end{definition}

\begin{definition}
\label{defPG} A \emph{mimp-graph} $G$ is a directed graph $\langle$\emph{V, E,
L}, $l_V$, $l_E\rangle$ where: \emph{V} is a set of nodes, \emph{E} is a set
of \emph{edges}, \emph{L} is a set of labels, $\langle v\in$ \emph{V}, $t
\in $ \emph{L}, $v^{\prime}\in $ \emph{V}$\rangle$, where $v$ is the source
and $v^{\prime}$ the target, $l_V$ is a labeling function from \emph{V} to 
\emph{R$\cup$F-Labels}, $l_E$ is a labeling function from \emph{E} to \emph{
E-Labels}.

 
Mimp-graphs are defined recursively as follows:
\begin{description}
\item[Basis]  If $G_1$ is a formula graph with root node $\alpha_m$~%
\footnote{%
We will use the terms $\alpha_m$, $\beta_n$ and $\gamma_r$ to represent the
principal connective of the formula $\alpha$, $\beta$ and $\gamma$
respectively.}, then the graph $G_2$ is defined as $G_1$ with the delimiter
nodes $H_n$ and $C$ and the edges $(\alpha_m, conc, C)$ and $(H_n, hyp,
\alpha_m)$ is a mimp-graph.
	
\item[\ $\to$E]  If $G_1$ and $G_2$ are mimp-graphs, and the graph (intermediate step) 
obtained by  $G_1 \oplus G_2$~%
\footnote{%
By definition $G_1 \oplus G_2$ equalizes the nodes of $G_1$ with the nodes
of $G_2$ that have the same label, and equalizes edges with the same source,
target and label into one. } contains the edge $(\to_q, l,\alpha_m)$ and the
two nodes $\to_q$ and $\alpha_m$ linked to the delimiter node $C$, then the
graph $G_3$ is defined as $G_1 \oplus G_2$ with
	
\begin{enumerate}
\item  the removal of the ingoing edges in the node $C$ which were generated 
in the intermediate step (see Figure~\ref{rules1}, dotted area in $G_1 \oplus G_2$);

\item  a rule node $\to$E$_i$ at the top level;

\item  the edges: $(\alpha_m, m_{new},\to$E$_i)$, $(\to_q$, $M_{new}$,$\to$E$_i)$, 
	$(\to$E$_i, c_{new}, \beta_n)$ and $(\beta_n, conc, C)$, where $new$ is a fresh (new) 
	index considering all edges of kind $c$, $m$ and $M$ ingoing and/or outgoing the formula-nodes
        $\alpha_m$, $\beta_n$ and $\to_q$,
\end{enumerate}
is a mimp-graph (see Figure~\ref{rules1}).

\begin{figure}[h] 
\centering
	\includegraphics [scale=0.7]{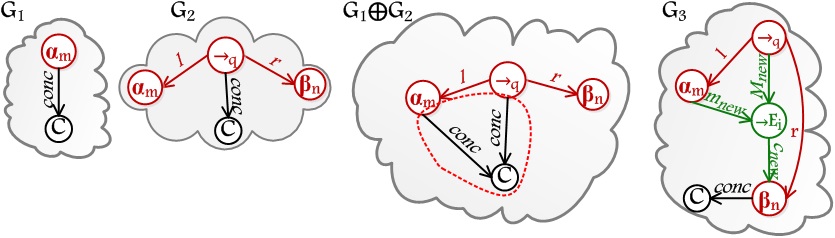}
	\caption{The $\to$E rule of mimp-graph}
	\label{rules1}
\end{figure}

\item[\ $\to$I\ ]  If $G_1$ is a mimp-graph and contains a node $\beta_n$
linked to the delimiter node $C$ and the node $\alpha_m$ linked to the
delimiter node $H_k$, then the graph $G_2$ is defined as $G_1$ with
	
\begin{enumerate}
\item  the removal of the edges: $(\beta_n, conc, C)$; 

\item  a rule node $\to$I$_j$ at the top level;

\item  a formula node $\to_t$ linked to the delimiter node $C$ by an edge $%
(\to_t, conc, C)$;

\item  the edges: $(\to_t, l, \alpha_m)$, $(\to_t, r, \beta_n)$, 
	$(\beta_n, p_{new},\to$I$_j$), $(\to$I$_j$, $c_{new}, \to_t)$, and $(\to$I$_j$, $disc_{new}, H_k)$, where $new$ is a fresh index concerning ingoing and outgoing edges of type $c$ and $p$ of the formula-nodes $\beta_n$, $\to_t$ and $\alpha_m$,
\end{enumerate}
is a mimp-graph (see Figure~\ref{rules2}; the $\alpha_m$-node is \emph{discharged}).
		
\item[\ $\to$I-v]  \footnote{%
the ``v'' stands for ``vacuous'', this case of the rule $\to$I discharges a
hypothesis vacuously. This means that $\alpha_m$ has no ingoing Hyp-edge}
If $G_1$ is a mimp-graph, and $G$ is a formula graph with root node $%
\alpha_m $, and $G_1$ contains a node $\beta_n$ linked to the delimiter node 
$C$, then the graph $G_2$ is defined as $G_1 \oplus G$ with

\begin{enumerate}
\item  the removal of the edge: $(\beta_n, conc, C)$;

\item  a rule node $\to$I$_j$ at the top level;

\item  a formula node $\to_t$ linked to the delimiter node $C$ by an edge $%
(\to_t, conc, C)$;

\item  the edges: $(\to_t, l, \alpha_m)$, $(\to_t, r, \beta_n)$, 
	$(\beta_n, p_{new},\to$I$_j$) and $(\to$I$_j$, $c_{new}, \to_t)$,
	where $new$ is an index under the same conditions of the previous case,
	\end{enumerate}
is a mimp-graph (see Figure~\ref{rules3}).

\end{description}
\end{definition}

\begin{figure}[h] 
\centering
	\includegraphics [scale=0.7]{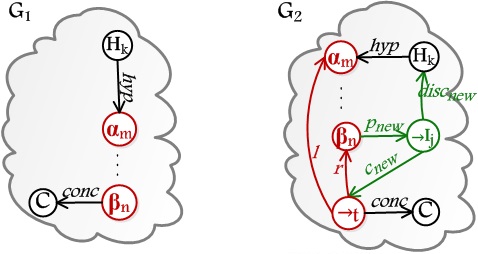}
	\caption{The $\to$I rule of mimp-graph}
	\label{rules2}
\end{figure}
\begin{figure}[h] 
\centering
	\includegraphics [scale=0.7]{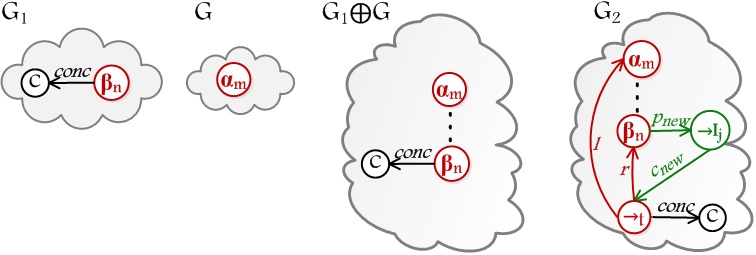}
	\caption{The $\to$I-v rule of mimp-graph}
	\label{rules3}
\end{figure}

Lemma\ref{lem-PG} enables us to prove that a given graph $G$ is a mimp-graph
without explicitly supplying a construction. Among others it says that we have 
to check that each node of $G$ is of one of the possible types that generate 
the Basis, $\to$E, $\to$I and $\to$I-v construction
cases of Definition~\ref{defPG}.

\begin{definition}[Inferential Ordering]
Let $G$ be a mimp-graph. An inferential order $>$ on nodes of $G$ is a partial
ordering of the rule nodes of $G$, such that, $n < n^{\prime}$, iff, $n$ and $%
n^{\prime}$ are rule nodes, and there is a formula node $f$, such that, n $%
\xrightarrow{l_1}$ f $\xrightarrow{l_2} n^{\prime}$ and $l_1$ is $c$ and $%
l_2 $ is $m$, or , $l_1$ is $c$ and $l_2$ is $M$, or, $l_1$ is $c$ and $l_2$
is $p$.
\end{definition} 

In order to avoid overloading of indexes, we will omit whenever is possible, 
the indexing of edges of kind $c$, $m$, $M$, $p$ and $disc$, remembering that 
the coherence of indexing is established by the kind of rule-node to which they 
are linked.

\begin{lemma}
$G$ is a mimp-graph if and only if the following hold: \label{lem-PG}

\begin{enumerate}
\item  There exists a well-founded (hence acyclic) inferential order $>$ on
all rule nodes of the mimp-graph\footnote{%
We can extend this ``green'' inferential order $>$ to the full ``mixed''
order $>^*$ by adding new ``red'' relations $>$ corresponding to arrows $%
\xrightarrow{l}$ and $\xrightarrow{r}$ between formula nodes. Note that $>^*$
may contain cycles (see Figure~\ref{example1}). However all recursive
definitions and inductive proofs to follow are based on the well-founded
``green'' order $>$, hence being legitimate.}.

\item  Every node $N$ of $G$ is of one of the following six types:

\begin{description}
\item[L \ \ \ ]  $N$ is labelled with one of the propositional letters: \{p,
q, r, ... \}. $N$ has no outgoing edges $l$ and $r$.

\item[F \ \ \ ]  $N$ has label $\to_n$ and has exactly two outgoing edges
with label $l$ and $r$, respectively. $N$ has outgoing edges with labels $p$, 
$m$ or $M$; and it has at most one ingoing edge with label $c$ and at
most one ingoing edge with label $hyp$.

\item[E \ \ \ ]  $N$ has label $\to$E$_i$ and has exactly one outgoing edge 
($\to$E$_i$, $c$, $\beta_n$), where $\beta_n$ is a node type \emph{L} or 
\emph{F}. $N$ has exactly two ingoing edges ($\alpha_m, m, \beta_n$) and 
($\to_q, M, \to$E$_i$), where $\alpha_m$ is a node type \emph{L} or \emph{F}.
There are two outgoing edges from the node $\to_q$: $(\to_q, l,\alpha_m)$
and $(\to_q, r, \beta_n)$.

\item[I \ \ \ ]  $N$ has label $\to$I$_j$ (or  $\to$I-v$_j$, if discharges an hypothesis vacuously), has one outgoing edge ($\to$I$_j$%
, $c$, $\to_t$), and one (or zero for the case $\to$I-v) outgoing edge ($\to$I$_j$, $disc$, $H_k$). 
$N$ has exactly one ingoing edge: $(\beta_n, p, \to$I$_j)$, where $\beta_n$
is a node type \emph{L} or \emph{F}. There are two outgoing edges from the
node $\to_t$: $(\to_t, l, \alpha_m)$ and $(\to_t, r, \beta_n)$.

\item[H \ \ \ ]  $N$ has label \emph{$H_k$} and has exactly one outgoing
edge $hyp$.

\item[C \ \ \ ]  $N$ has label \emph{C} and has exactly one ingoing edge $%
conc$.
\end{description}
\end{enumerate}
\end{lemma}

\begin{proof}

$\Rightarrow$: Argue by induction on the construction of mimp-graph (Definition~\ref{defPG}). 
For every construction case for mimp-graphs we have to check the three properties 
stated in Lemma. Property (2) is immediate. For property (1), we know 
from the induction hypothesis that there is an inferential order $>$ on rule 
nodes of the mimp-graph. In the construction cases $\to$I, $\to$I-v or $\to$E, we make the new rule node that is introduced highest in the $>$-ordering, which yields an 
inferential ordering on rule nodes. 
In the construction case $\to$E, when we have two inferential 
orderings, $>_1$ on $G_1$ and $>_2$ on $G_2$. Then $G_1 \oplus G_2$ 
can be given an inferential ordering by taking the union of $>_1$ and $>_2$ 
and in addition putting $n > m$ for every rule node $n, m$ such that  
$n \in G_1, m \in G_2$.

$\Leftarrow$: Argue by induction on the number of rule nodes of $G$. 
Let $>$ be the topological order that is assumed to exist. Let $n$ 
be the rule node that is maximal w.r.t. $>$. Then $n$ must be on 
the top position. When we remove node $n$, including its edges 
linked (if $n$ is of type I) and the node type $C$  is linked to 
the premise of the rule node, we obtain a graph $G'$ that satisfies 
the properties listed in Lemma. By induction hypothesis we see 
that $G'$ is a mimp-graph. Now we can add the node $n$ again, using one of 
the construction cases for mimp-graphs: {\em Basis} if $n$ is a 
{\em L} node or {\em F} node, $\to$E if $n$ is an {\em E} node, $\to$I if $n$ is an {\em I} node.
\end{proof}

It is natural to consider minimal mimp-graph-like representations of given
natural deductions. Actually one can try to minimize the number of F-Labels
and/or R-Labels, but for the sake of brevity we consider only the F-option,
as it helps to reduce the size under standard normalization (see the next
section). To grasp the point note that mimp-graph in Figure~\ref{example1}
(see above) is F-minimal, i.e. its F-labelled nodes refer to pairwise
distinct formulas. This observation is summarized by

\begin{theorem}[F-minimal representation]\label{theo1}
Every standard tree-like natural deduction $\Pi $ has a uniquely
determined (up to graph-isomorphism) \emph{F-minimal mimp-like representation%
} $G_{\Pi }$, i.e. such a one that satisfies the following four conditions.

\begin{enumerate}
\item  $G_{\Pi }$ is a mimp-graph whose size does not exceed the size of $%
\Pi $.

\item  $\Pi $ and $G_{\Pi }$ both have the same (set of) hypotheses and the
same conclusion.

\item  There is graph homomorphism $h:\Pi \rightarrow G_{\Pi }$ that is
injective on R-Labels.

\item  All F-Labels occurring in $G_{\Pi }$\ denote pairwise distinct
formulas.
\end{enumerate}
\end{theorem}

\begin{proof}
Let $N$ and $F$ be the set of nodes and formulas, respectively, occurring in 
$\Pi $. Note that $\Pi $ determines a fixed surjection $f:N\rightarrow F$
that may not be injective (for in $\Pi $, one and the same formula may be
assigned to different nodes). In order to obtain $G_{\Pi }$ take as R-nodes
the inferences occurring in $\Pi $ assigned with the corresponding ``green''
R-Labels representing inferences' names (possibly indexed, in order to
achieve a 1--1 correspondence between inferences and R-Labels, cf. Figure~\ref{example1}).\ 
Define basic F-nodes of $G_{\Pi }$ as formulas from $F$\ assigned with
the corresponding ``red'' F-Labels representing formulas' principal
connectives (possibly indexed, in order to achieve a 1--1 correspondence
between formulas and F-Labels, cf. Figure~\ref{example1}). So the total number of all
basic F-nodes of $G_{\Pi }$ is the cardinality of the set $F$, while $f$
being a mapping from the nodes of $\Pi $ onto the basic F-nodes of $G_{\Pi }$%
. To complete the construction of $G_{\Pi }$ we add, if necessary, the
remaining F-nodes labelled by failing ``red'' representations of subformulas
of $f(x)$, $x\in N$, and define the E-Labels of $G_{\Pi }$ (both ``green''
and ``red''), accordingly. Note that by the definition all nodes of $G_{\Pi }
$ have pairwise distinct labels. In particular, every F-Label occurs only
once in $G_{\Pi }$, which yields the crucial condition 4.
\end{proof}

\section{Normalization for mimp-graphs}\label{sec:normal}

In this section we define the normalization procedure for mimp-graphs. It is
based on standard normalization method given by Prawitz. Thus a \emph{%
maximal formula} in mimp-graphs is a $\rightarrow $-$I$ followed by a $%
\rightarrow $-E of the same formula graph (see Definition~\ref{maxformula}).
It is the same notion of maximal formulas that is being used in natural
deduction derivations. So a maximal formula occurrence is the consequence of
an application of an introduction rule and major premise of an application
of an elimination rule. But here we assume that derivations are represented
by mimp-graphs. We wish to eliminate such maximal formula by dropping nodes
and edges that are involved in the maximal formula. However, it could also
happen that between the rule nodes $\rightarrow $-I and $\rightarrow $-E
there are several other maximal formulas. 

\begin{definition}
\label{maxformula} A \emph{maximal formula} $m$ in a mimp-graph $G$ (see
Figure~\ref{maxforfig}) is a sub-graph of $G$ consisting of:

\begin{enumerate}
\item  the formula nodes $\alpha_m$, $\beta_n$, $\to_q$, the rule node $\to$I%
$_i$ and the delimiter node $H_u$;

\item  the rule node $\to$E$_j$ at the top level;

\item  the edges: $(\to_q, l, \alpha_m)$, $(\to_q,r, \beta_n)$, $(\beta_n,
p, \to$I$_i$), $(\to$I$_i, c, \to_q)$, $(H_u, hyp, \alpha_m)$, $(\to$I$_i,
disc, H_u)$, $(\alpha_m,$ $m,\to$E$_j)$, $(\to_q, M, \to$E$_j)$ and $(\to$E$%
_j, c, \beta_n)$;
\end{enumerate}
\end{definition}

\begin{figure}[h]
\begin{center}
\begin{minipage}{6cm}

 $\begin{array}{lcr}
  \a{$\Pi_1$}
  \noLine
  \u{$\alpha$} 
  \a{$[\alpha]^u$}
  \noLine
  \u{$\Pi_2$}
  \noLine
  \u{$\beta$}\RightLabel{$u$}
  \u{$\alpha\to\beta$}
  \b{$\beta$} 
  \noLine
  \u{$\Pi_3$}
\DisplayProof
\hspace{1cm}\Rightarrow

\end{array}$  
\end{minipage}
\begin{minipage}{7cm}
\includegraphics[scale=0.55]{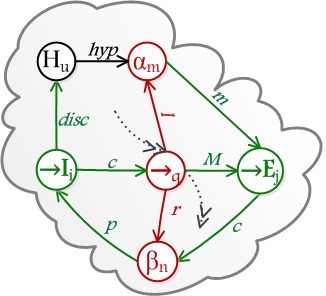}
\end{minipage}
\caption{Maximal formula in mimp-graphs}
\label{maxforfig}

\end{center}
\end{figure}

\begin{definition}
\label{defbranch} (1) For $n_i \in V$, a \emph{p-path} in a proof-graph is a
sequence of vertices and edges of the form: $n_1 \xrightarrow{l_1} n_2 %
\xrightarrow{l_2}... \xrightarrow{l_{k-2}} n_{k-1} \xrightarrow{l_{k-1}}
n_{k}$, such that $n_1$ is a hypothesis formula node, $n_k$ is the
conclusion formula node, $n_i$ alternating between a rule node and a formula
node. The edges $l_i$ alternate between two types of edges: the first is $%
l_j \in \{m, M, p\}$ and the second $l_j=c$. (2) A \emph{\ branch} is an
initial part of a \emph{p-path} which stops at the conclusion formula node
or at the first minor premise whose major premise is the conclusion of a
rule node.
\end{definition}

\begin{definition}
\label{defReorder} Let $G$ a graph obtained by dropping rule nodes in a
mimp-graph then, the reordering of $G$ is defined as the graph $G$ with the
following (new) inference order on the rule nodes of $G$.


\begin{itemize}
\item  $o(t_m) = 0$ for a rule node $t_m$ starting with hypothesis.

\item  $o(t) = o(t^{\prime}) + 1$ if the conclusion formula of rule node $%
t^{\prime}$ is premise or major premise of $t$.
\end{itemize}
\end{definition}

\begin{proposition}
\label{prop} If a graph $G$ is obtained by a reordering by means of the
operation defined in Definition~\ref{defReorder} then, $G$ is a mimp-graph.
\end{proposition}

\begin{definition}
\label{elimina} Given a mimp-graph $G$ with a maximal formula $m$,
eliminating a maximal formula is the following transformation of a
mimp-graph, where the maximal formula $m$ satisfies the following
requirements:

\begin{enumerate}
\item  Between the rule nodes $\to$I$_i$ and $\to$E$_l$ there are zero or 
more maximal formulas with inferential orders within the range of these rule
nodes.
\item  There is an edge $(\to$I$_i ,c, \to_q)$, and, the formula node $\to_q$
has zero  or more ingoing edges.

\item  There is an edge $(\to_q, M, \to$E$_l)$, and, the formula node $\to_q$
is  the premise of zero or more of another rule nodes.

\item  If a branch will be separated from the inferential order this branch 
must be insertable in the following branch, according to the order,  i.e.
the conclusion of this separated branch is the premise in the following
branch.
\end{enumerate}

The elimination of a maximal formula is the following operation on a
mimp-graph  (see Figure~\ref{case2}, the dotted arrows are representing sets
of edges):

\begin{enumerate}
\item  If there is no maximal formula between the rule nodes $\rightarrow $I$%
_{i}$ and $\rightarrow $E$_{l}$ then follow these steps:
\begin{enumerate}
\item  If the edge $(\to$I$_i, c, \to_q)$ is the only ingoing edge to $\to_q
$ and  the edge $(\to_q, M, \to$E$_l)$ is the only outgoing edge from $\to_q$
then remove the edges to and from the formula node $\to_q$, and the  formula
node $\to_q$.

\item  Remove the edges to and from the nodes $\to$I$_i$, $\to$E$_l$ and $H_u
$.

\item  Remove the nodes $\to$I$_i$, $\to$E$_l$ and $H_u$.

\item  Apply the operation defined in Definition~\ref{defReorder} to the
resulting graph. Note that Proposition~\ref{prop} ensures that the result is
a mimp-graph.
\end{enumerate}

\item  Otherwise eliminate the maximal formulas between the rule nodes $\to$I$_i$
and $\to$E$_l$ as in the previous step. 
\end{enumerate}

\begin{figure}[h]
\centering
\includegraphics [scale=0.55]{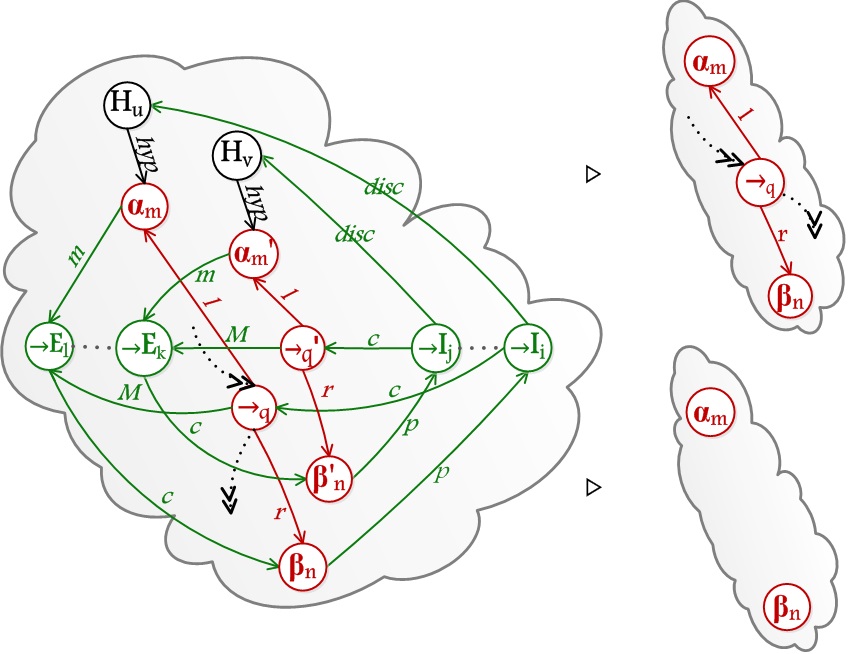}
\caption{Elimination of a maximal formula in mimp-graphs}
\label{case2}
\end{figure}
\end{definition}

Note that the removal of a node $\to$I generated by case $\to$I-v, in the
Definition~\ref{defPG}, disconnects the graph meaning that the sub-graph
hypotheses linked, by the edge $m$, to eliminated node $\to$-E is no longer
connected to the delimiter $C$.

\begin{figure}[p]
\begin{minipage}[c]{15cm}

\centering
\begin{small}

 \centering
 $\begin{array}{lcr}
 \a{$\Pi_2$}
  \noLine
  \u{$\beta$} 

  \a{$\Pi_1$}
  \noLine
  \u{$\alpha$} 
 
  \a{$[\beta]^v [\alpha]^u$}
  \noLine
  \u{$\Pi_0$}
  \noLine
  \u{$\gamma$}\RightLabel{$v$}
  \u{$\beta\to\gamma$}\RightLabel{$u$}
  \u{$\alpha\to(\beta\to\gamma)$}

  \b{$\beta\to\gamma$} 

  \b{$\gamma$} 
\DisplayProof
 &
\rhd
\hspace{0.5cm}
  \a{$\Pi_2$} 
  \noLine
  \u{$\beta$}
  \a{$[\beta]^v$}
  \a{$\Pi_1$}
  \noLine
  \u{$\alpha$}
  \noLine
  \b{$\Pi_0$}
  \noLine
  \u{$\gamma$}
  \u{$\beta \to \gamma$}
  \b{$\gamma$}
\DisplayProof
\hspace{0.5cm}
\rhd
&
  \a{$\Pi_2$ \ $\Pi_1$}
    \noLine
  \u{$\beta$ \ \ \ \ $\alpha$}
  \noLine
  \u{$\Pi_0$}
  \noLine
  \u{$\gamma$}
\DisplayProof

 \end{array}$
\end{small}	
	
	\vspace{0.8cm}
 	$\Downarrow trans$
	 \vspace{0.8cm}
\end{minipage}
\begin{minipage}[r]{15cm}
	\centering
	\includegraphics[scale=0.47]{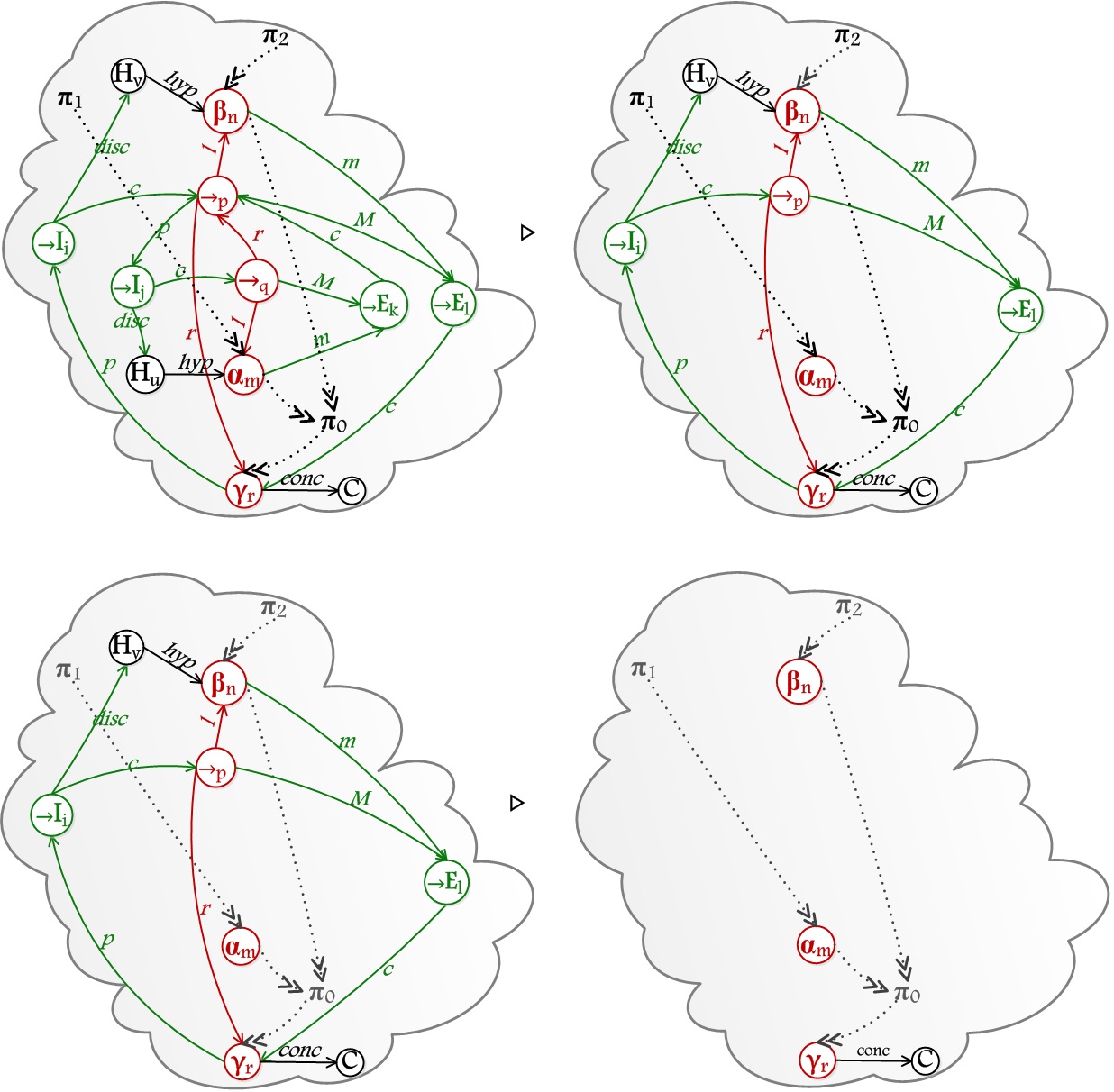}
\end{minipage}
\caption{Eliminating a maximal formula in a natural deduction proof and its mimp-graph translation}
\label{case3i}
\end{figure}

Let us show in Figure~\ref{case3i} an instance of the eliminating a maximal
formula in tree form. Note that this example shows the reason why
essentially our (weak) normalization theorem is directly a strong
normalization theorem. The formula $\beta\to\gamma$ is not a maximal formula
before a reduction is applied to eliminate the maximal formula $%
\alpha\to(\beta\to\gamma)$. This possibility of having hidden maximal
formulas in Natural Deduction is the main reason to use more sophisticated
methods whenever proving strong normalization. In mimp-graphs there is no
possibility to hide a maximal formula because all formulas are represented
only once in the graph. In this graph $\beta\to\gamma$ is already a maximal
formula. We can choose to remove any of the two maximal formulas. If $%
\beta\to\gamma$ is chosen to be eliminated, by the mimp-graph normalization
procedure, its reduction eliminates the $\alpha\to(\beta\to\gamma)$ too. On
the other hand, the choice of $\alpha\to(\beta\to\gamma)$ to be reduced only
eliminates itself. In any case the number of maximal formula decreases.

We shall construct the normalization proof for mimp-graphs. This proof is
guided by the normalization measure. That is, the general mechanism from the
proof determines that a given mimp-graph $G$ should be transformed into a
non-redundant mimp-graph by applying of reduction steps and at each
reduction step the measure must be decreased. The normalization measure will
be the number of maximal formulas in the mimp-graph.

Also note a following important observation concerning F-minimal mimp
representations (see Theorem \ref{theo1}). Since F-minimal mimps can have at most one
occurrence of hypotheses $\alpha$ and/or $\beta$, every proper reduction
step will diminish the size of deduction. Hence the size of the graph (= the
number of nodes) can serve as another inductive parameter, provided that the
normalization is being applied to F-minimal mimp-graph representations.

\begin{theorem}[Normalization]\label{theo2}
Every mimp-graph $G$ can be reduced to a normal mimp-graph $G^{\prime }$
having the same hypotheses and conclusion as $G$. Moreover, for any standard
tree-like natural deduction $\Pi$, if $G:=G_{\Pi}$ (the F-minimal
mimp-like representation of $\Pi$, cf. Theorem \ref{theo1}), then the size of $%
G^{\prime }$ does not exceed the size of $G$, and hence also $\Pi $.
\end{theorem}

\begin{remark}
The second assertion sharply contrasts to the well-known exponential speed-up of
standard normalization. Note that the latter is a consequence of the
tree-like structure of standard deductions having different occurrences of
equal hypotheses formulas, whereas all formulas occurring in F-minimal
mimp-like representations are pairwise distinct.
\end{remark}

\begin{proof}
This characteristic of preservation of the premises and conclusions of the 
derivation is proved naturally. Through an inspection of each elimination 
of maximal formula is observed that the reduction step (see Definition~\ref{elimina}) 
of the mimp-graph does not change the set of premises and conclusions 
(indicated by the delimiter nodes $H$ and $C$) of the derivation that is being reduced.

In addition, the demonstration of this theorem has two primary requirements. 
First, we guarantee that through the elimination of maximal formulas in the 
mimp-graph, cannot generate more maximal formulas.
The second requirement is to guarantee that during the normalization process, 
the normalization measure adopted is always reduced. 

The first requirement is easily verifiable through an inspection of each case 
in the elimination of maximal formulas. Thus, it is observed that no case 
produces more maximal formulas. The second requirement is established through 
the normalization procedure and demonstrated through an analysis of existing 
cases in the elimination of maximal formulas in mimp-graphs. To support this 
statement, it is used the notion of normalization measure, we adopt as measure 
of complexity (induction parameter) the number of maximal formulas $Nmax(G)$. 
Besides, as already mentioned, working with F-mimimal mimp-graph representations we 
can use as optional inductive parameter the ordinary size of mimp-graphs.   
\end{proof}

\textbf{Normalization Process}

We know that a specific mimp-graph $G$ can have one or more maximal formulas
represented by $M_1, ..., M_n$. Thus, the normalization procedure is
described by the following steps:

\begin{enumerate}
\item  Choosing a maximal formula represented by $M_k$.

\item  Identify the respective number of maximal formulas $Nmax(G)$.

\item  Eliminate maximal formula $M_k$ as defined in Definition~\ref{elimina}.

\item  In this application one of the following three cases may occur:

\begin{description}
\item[a)]  The maximal formula is removed.
\item[b)]  The maximal formula is removed but the formula node is
maintained,  hence $Nmax(G)$ is decreased;
\item[c)]  All maximal formulas are removed.
\end{description}

\item  We repeat this process until the normalization measure $Nmax$ is
reduced  to zero and $G$ becomes a normal mimp-graph.
\end{enumerate}

Since the process of the eliminating a maximal formula on mimp-graphs always
ends in the elimination of at least one maximal formula, 
and with the decrease in the number of vertices of the graph, we can say that 
this normalization theorem is directly a strong normalization theorem. 

\section{Conclusions and Related works}

This representation of a proof in mimp-graph requires fewer nodes than the 
tree and the list representation of proofs. For the case of lists, it is enough to 
observe that a sub-formula of a formula is already in any graph representation of it. 
If both take part in the proof the size is smaller than in the mentioned 
representations. The ability to represent any Natural Deduction proof is preserved.
Another important advantage of a compact representation of graphs is that it 
allows to deduce some structural properties of proof-graphs, for example based on 
a mimp-graph, it is easy to see an upper bound in the length of the reduction sequence 
to obtain a normal proof. It is the number of maximal formulas.

There is some previous research concerning the use of graphs to represent proofs 
developed on connections to substructural logics as Linear Logic, see \cite{Girard1} 
and \cite{Girard2} for example. The main motivation of this just mentioned 
investigations is to provide a sound way of representing Linear Logic proofs without 
dealing with unique labeling and complicated rules for relabeling and discharging 
mechanisms need to represent Linear Logic proofs as trees in Natural Deduction styles 
as well as in Sequent Calculus. 
Proof-nets were such representations and a syntactical criteria on the possible paths 
on them were considered as a soundness criteria for a proof-graph to be a proof-net. 
Proof-nets have a cut-rule quite similar to the cut in Sequent Calculus. 
For the Multiplicative fragment of Classical Linear Logic, there is a linear time 
cut-elimination theorem. However, when the additive versions of the connectives are 
considered, the usual complexity of the cut-elimination raises up again. Linear Logic 
is an important Logic whenever we consider the study of a concurrent computational 
systems and its semantics strongly uses concurrency theory concepts. Our investigations, 
on the other hand, is not motivated by proof-theoretical semantics~\footnote{%
The name that nowadays it is used to denote the kind of research pioneered by Jean-Yves Girard}. 
From the purely proof-theoretical point of view, we use graphs to reduce the redundancy 
in proofs in such a way that we do not allow hidden maximal formulas in our graph 
representation of a Natural Deduction proof. As a secondary motivation, this is a 
preliminary step into investigating how a theorem prover based on graphs is more 
efficient than usual theorem provers. Our proof-graphs represent the formulas 
themselves in a way that each subformula is a unique node in the graph. 
Proof nets do not represent formulas in this way.

\bibliographystyle{eptcs}
\bibliography{ref}



\end{document}